\documentclass[a4paper]{amsart}

\usepackage{fullpage}

\usepackage[utf8]{inputenc}
\usepackage[T1]{fontenc}
\usepackage[english]{babel}

\usepackage[all]{xy}
\usepackage{qsymbols,xparse}
\usepackage{amsmath,amstext,amsthm}
\usepackage{graphicx,caption,subcaption}

\usepackage{url}
\usepackage{breakurl} 
\usepackage[breaklinks]{hyperref}
\usepackage{nicefrac}

\theoremstyle{definition}

\newtheorem{prop}{Proposition}

\newtheorem{theorem}{Theorem}
\newtheorem{corollary}{Corollary}
\newtheorem{example}{Example}

\newcommand{\set}[1]{\{#1\}}
\newcommand{\seq}[1]{\left(#1\right)}
\renewcommand{\zerodot}{\mbox{$0\hspace*{-4pt}\raisebox{.5pt}{$\cdot$}$\,}}
\newcommand{\domsing}[1]{`r_{#1}}
\newcommand{\nthcoeff}[1]{[z^n]{#1}}
\newcommand{\D}{\mathcal{D}}				
\newcommand{\M}{\mathcal{M}}				
\newcommand{\N}{\mathcal{N}}				
\newcommand{\calH}{\mathcal{H}}			
\newcommand{\K}{\mathcal{K}}				

\newcommand{\ltobw}{\ensuremath{\mathsf{LtoBw}}}
\newcommand{\bwtol}{\ensuremath{\mathsf{BwtoL}}}
\newcommand{\bwtobz}{\ensuremath{\mathsf{BwtoBz}}}
\newcommand{\bztobw}{\ensuremath{\mathsf{BztoBw}}}

\newcommand{\untol}{\ensuremath{\mathsf{UnToL}}}
\newcommand{\untod}{\ensuremath{\mathsf{UnToD}}}
\newcommand{\ltoun}{\ensuremath{\mathsf{LToUn}}}
\newcommand{\dtoun}{\ensuremath{\mathsf{DToUn}}}
\newcommand{\motone}{\ensuremath{\mathsf{MoToNe}}}
\newcommand{\netomo}{\ensuremath{\mathsf{NeToMo}}}

\newcommand{\upath}[2]{
    \xymatrix @=.6pc {
        #1 \ar@{-}[d]\\
        #2
    }
}

\newcommand{\bcross}[3]{
    \xymatrix @=.9pc {
        & #1 \ar@{-}[d] & \\
        & \bullet \ar@{-}[dl] \ar@{-}[dr] & \\
        #2 & & #3
    }
}

\newcommand{\btcross}[4]{
    \xymatrix @=.9pc {
            & #1 \ar@{-}[dl] \ar@{-}[dr] & \\
            #2 & & #3 \ar@{-}[d] \\
            & & #4
        }
}

\newcommand{\bnode}[3]{
    \xymatrix @=.9pc {
            & #1 \ar@{-}[dl] \ar@{-}[dr] & \\
            #2 & & #3
        }
}

\newcommand{\xdownarrow}[1]{%
  {\left\downarrow\vbox to #1{}\right.\kern-\nulldelimiterspace}
}

\newcommand{\nats}{\mathbb{N}}
\newcommand{\C}{\mathbb{C}}
\newcommand{\Zet}{\mathbb{Z}}
\newcommand{\card}[1]{|#1|}

\DeclareDocumentCommand{\gL}{O{\infty}}{L_{#1}(z)}
\DeclareDocumentCommand{\cgL}{O{\infty}}{\mathcal{L}_{#1}}

\DeclareDocumentCommand{\agL}{O{\infty}}{M_{#1}(z)}
\DeclareDocumentCommand{\acgL}{O{\infty}}{\mathcal{M}_{#1}}

\newcommand{\DeltagL}{\Delta_{\gL}}
\newcommand{\DeltagLm}{\Delta_{\gL[m]}}

\newcommand{\B}{\mathcal{B}}
\newcommand{\W}{\mathcal{W}}
\newcommand{\Z}{\mathcal{Z}}
\newcommand{\T}{\mathcal{T}}

\newcommand{\lf}[2]{{\raisebox{.8pc}{\tiny
\xymatrix@=.05pc{&{#2}\ar@{-}[dl]\\{#1}}}}}

\newcommand{\ri}[2]{{\raisebox{.8pc}{\tiny
\xymatrix@=.05pc{{#1}\ar@{-}[dr]\\&{#2}}}}}

\title{Combinatorics of $\lambda$-terms: a natural approach}

\author{Maciej Bendkowski}
\address{Jagiellonian University\\
        Faculty of Mathematics and Computer Science\\
        Theoretical Computer Science Department\\
        ul. Prof.~{\L}ojasiewicza 6, 30--348 Krak\'ow, Poland}
\email{\set{bendkowski,grygiel,zaionc}@tcs.uj.edu.pl}
\author{Katarzyna Grygiel}
\author{Pierre Lescanne}
\address{University of Lyon\\
        \'Ecole normale sup\'erieure de Lyon\\
        LIP (UMR 5668 CNRS ENS Lyon UCBL INRIA)\\
        46 all\'ee d'Italie, 69364 Lyon, France}
\email{pierre.lescanne@ens-lyon.fr}
\author{Marek Zaionc}

\date{\today}

\thanks{This work was partially supported by the Polish National Science Center grant 2013/11/B/ST6/00975.} 
            

\begin{document}

	\begin{abstract}
	We consider combinatorial aspects of $`l$-terms in the model based on
        de~Bruijn indices where each building constructor is of size
        one. Surprisingly, the counting sequence for $`l$-terms corresponds also to
        two families of binary trees, namely black-white trees and zigzag-free ones. We provide a constructive proof of this fact by exhibiting appropriate bijections. Moreover, we identify the sequence of Motzkin numbers with the counting sequence for neutral $`l$-terms, giving a bijection which, in consequence, results in an exact-size sampler for the latter based on the exact-size sampler for Motzkin trees of Bodini et alli. Using the powerful theory of analytic combinatorics, we state several results concerning the asymptotic growth rate of $`l$-terms in neutral, normal, and head normal forms. Finally, we investigate the asymptotic density of $`l$-terms containing arbitrary fixed subterms showing that, inter alia, strongly normalising or typeable terms are asymptotically negligible in the set of all $`l$-terms.
	\end{abstract}

    \maketitle

    \section{Introduction}\label{sec:introduction}

	Quantitative investigations in computational logic, where combinatorial aspects and asymptotic behaviour of large typical entities related to computations and logic are studied, form an attractive and actively developed branch of modern computer science. The unique combination of combinatorics, logic, and computer science leads to a synthesis of approaches and techniques resulting in new discoveries regarding the relation between computations and their syntactic realisation.
	
	Representing a rather functional approach to logic and computations, lambda calculus was first studied by David et al. (see~\cite{david}). Assuming a canonical representation of closed  $`l$-terms, David et al. showed that asymptotically almost all $`l$-terms are strongly normalising. Similarly to their model, the authors of~\cite{gittenberger-2011-ltbuh} considered the model in which the size of every variable, application, and abstraction is equal to one. A different model of lambda calculus with de~Bruijn indices, used to cope with the infinite number of variables, was considered in~\cite{gry_les}. In~\cite{tromp} John Tromp introduced a binary encoding of lambda calculus and combinatory logic, which allowed him to construct compact and efficient self-interpreters for both languages. Quantitative aspects of the aforementioned lambda calculus representation were studied in~\cite{JFP:10090684}. The framework of combinatory logic, were computations are represented without the use of bound variables, was investigated in~\cite{david,Bendkowski2015}.

    It is worth noticing that $`l$-calculus and combinatory logic are not the only computational models considered in the literature. In~\cite{Hamkins2006} Hamkins and Miasnikov considered the
    framework of Turing machines, showing that the halting problem is decidable on a
    set of asymptotic density one among the set of all Turing machines. Somewhat
    contrary to their result, Bienvenu et al. considered a different
    information-theory model of Turing machines, showing that the set of terminating
    Turing machines has no asymptotic density~\cite{Bienvenu2016}. In other
    words, the sequence of probabilities that a uniformly random Turing machine of
    size $n$ terminates, has no limit as $n$ tends to infinity.
    
    In this paper we propose a natural way of measuring the size of $`l$-terms
    represented using the unary de~Bruijn notation. In our model we assume that all the building constructors, i.e.~$`l$-abstractions, applications, successors and
    zeros contribute one to the term size.
    
    The paper is organised as follows. In~\autoref{sec:analytic-combinatorics} we
    list the employed analytic tools and generating function methods with
    corresponding notation. In~\autoref{sec:natural-counting} we state our natural
    combinatorial model. In~\autoref{subsec:plain-lambda-terms} we count plain
    $`l$-terms giving a corresponding holonomic equation in the
    subsequent~\autoref{subsec:linfinity-holonomic}. Next, we exhibit bijections
    between plain $`l$-terms and black-white trees (see
    sections~\ref{subsec:bw-trees} and~\ref{subsec:bijection-black-white}) as well as
    zigzag-free trees (see sections~\ref{subsec:zz-trees}
    and~\ref{subsec:bijection-zigzag}). In sections~\ref{subsec:nf-lambda-terms}
    and~\ref{subsec:bijection-motzkin-trees} we focus on neutral $`l$-terms and
    $\beta$-normal forms, exhibiting a bijection between the former terms and Motzkin
    trees. Head normal forms and neutral head normal forms are considered
    in~\autoref{subsec:head-normal-forms}. In the
    next~\autoref{subsec:m-open-lambda-terms} we count the number of plain $`l$-terms
    with bounded number of free indices. In~\autoref{sec:another-size-notions} we
    discuss some alternative notions of size. Finally,
    in~\autoref{sec:fixed-lambda-terms} we focus on the family of $`l$-terms
    containing any arbitrary fixed subterm, showing that in the considered model asymptotically almost all $`l$-terms are neither strongly normalising nor typeable.

	A conference version of this paper appeared as~\cite{Bendkowski2016}.

\section{Generating functions and analytic methods}\label{sec:analytic-combinatorics}
Suppose that we are given a countable set of objects $A$ and a size function
$f \colon A \to \nats$ such that $a_n := \card{f^{-1}(\set{n})}$ is finite for each
$n \in \nats$, i.e.~there are only finitely many objects of any given size $n$. We 
call the pair ${\mathcal A}=(A,f)$ a \emph{combinatorial class}. In such a case, 
we can consider $\mathcal{A}$'s counting sequence
${\seq{a_n}}_{n \in \nats}$ with its corresponding ordinary generating function
$A(z) = \sum_{n \geq 0} a_n z^n$.  Viewing $A(z)$ as an analytic function defined in
some neighbourhood of the complex plane origin, we can employ the methods of
\emph{analytic combinatorics} (see,
e.g.~\cite{Wilf:2006:GEN:1204575,Flajolet:2009:AC:1506267}) and link the properties
of $A(z)$ with the asymptotic behaviour of its underlying counting sequence
${\seq{a_n}}_{n \in \nats}$.

Throughout the paper, we use the following common notational conventions and
abbreviations. We use capital calligraphic letters
$\mathcal{A},\mathcal{B},\mathcal{C},\ldots$ to denote combinatorial classes. Their
corresponding ordinary generating functions are denoted as
$A(z),B(z),C(z),\ldots$. The coefficient standing by $z^n$ in the Maclaurin series
expansion of $A(z)$ is denoted as $[z^n]A(z)$. Whenever a generating function $A(z)$ yields
a dominating singularity, we use $\domsing{A}$ to denote it. Sometimes, when we are
interested in the approximate value of $\domsing{A}$ we write $\domsing{A} \doteq c$,
where $c$ is its numerical approximation. To denote addition and subtraction
operations on combinatorial classes we use $`(+)$ and $`(-)$, respectively. Given two 
sequences ${\seq{a_n}}_{n \in \nats}$ and ${\seq{b_n}}_{n \in \nats}$ of the same asymptotic 
order, i.e.~satisfying $\lim_{n\to\infty} \nicefrac{a_n}{b_n} = 1$, we write $a_n \sim b_n$. 
Since we are exclusively dealing with ordinary generating functions, whenever we write generating
function, we mean ordinary generating function. We use the underbar notation to denote de~Bruijn
indices. And so, $\underbar{\sf n}$ stands for the $n$th de~Bruijn index,
i.e.~$S^{n}\, \zerodot$.

\subsection{Analytic combinatorics} The employed method of \emph{singularity
  analysis}~(see \cite{Flajolet:2009:AC:1506267}) consists of a few steps. We start
with a combinatorial class $\mathcal{A}$. Then, we find a closed form expression
defining its generating function $A(z)$. Next, we locate $A(z)$'s dominant
singularities, i.e.~singularities with smallest modulus, determining the exponential
growth rate of $\seq{a_n}_{n \in \nats}$ as dictated by the following theorem.

\begin{theorem}[Exponential Growth Formula, see {\cite[Theorem IV.7]{Flajolet:2009:AC:1506267}}]\label{thm:exponential-growth-formula}
  If $A(z)$ is analytic at $0$ and $R$ is the modulus of a singularity nearest to the
  origin in the sense that
  \[ R = \sup \{ r \geq 0 ~:~ A(z) \text{ is analytic in } |z| < r \} ,\]
  then the coefficient $a_n = [z^n] A(z)$ satisfies
  \[ a_n = R^{-n} \theta(n) \quad \text{with} \quad \limsup |\theta(n)|^{\frac{1}{n}}
  = 1 .\]
\end{theorem}

In the case of analytic functions derived
from combinatorial classes, the search for dominant singularities simplifies to
finding $A(z)$'s radius of convergence.

\begin{theorem}[Pringsheim, see {\cite[Theorem~IV.6]{Flajolet:2009:AC:1506267}}]
  \label{th:pringsheim}
  If $A(z)$ is representable at the origin by a series expansion that has
  non-negative coefficients and radius of convergence $R$, then the point $z = R$ is
  a singularity of $A(z)$.
\end{theorem}

In order to find sub-exponential factors contributing to
${\seq{a_n}}_{n \in \nats}$'s growth rate, we have to determine the types and
relative location of $A(z)$'s dominating singularities. If $A(z)$ has just one single
algebraic dominating singularity, we can use the following standard function scale
combined with the well known Newton-Puiseux series expansion
(see~{\cite[Chapter~VI.4. The process of singularity
  analysis]{Flajolet:2009:AC:1506267}}).

\begin{theorem}[Standard function scale, see~{\cite[Theorem~VI.1]{Flajolet:2009:AC:1506267}}]\label{th:standard-func-scale}
  Let $\alpha \in \C \setminus \Zet_{\leq 0}$. Then $f(z) = {(1 - z)}^{-\alpha}$
  admits for large $n$ a complete asymptotic expansion in form of
  \begin{equation*} [z^n]f(z) = \frac{n^{\alpha-1}}{\Gamma(\alpha)} \left( 1 +
      \frac{\alpha(\alpha-1)}{2n} +
      \frac{\alpha(\alpha-1)(\alpha-2)(3\alpha-1)}{24n^2} + O(\frac{1}{n^3}) \right)
  \end{equation*}
  where $\Gamma$ is the Euler Gamma function.
\end{theorem}

\begin{theorem}[Newton-Puiseux, see~{\cite[Theorem~VII.7]{Flajolet:2009:AC:1506267}}]\label{th:newton-puiseux}
  Let $f(z)$ be a branch of an algebraic function $P(z, f(z)) = 0$. Then in a
  circular neighbourhood of a singularity $\zeta$ slit along a ray emanating from
  $\zeta$, $f(z)$ admits a fractional series expansion that is locally convergent and
  of the form
  \begin{equation*}
    f(z) = \sum_{k \geq k_0} c_k {\left( z - \zeta \right)}^{\nicefrac{k}{\kappa}},
  \end{equation*}
  where $k_0 \in \Zet$ and $\kappa \geq 1$.
\end{theorem}

\section{Natural counting}\label{sec:natural-counting}
Let $N$ and $M$ be two $`l$-terms with some bound variables. If bound variables in
$N$ can be renamed in such a way that $N$ and $M$ become equal, then both $N$ and $M$
are said to be $\alpha$-convertible. In particular, this is an equivalence relation
(see, e.g.~\cite{barendregt1984}) such that if two $`l$-terms belong to the same
$\alpha$-equivalence class, then both represent the same computable function (though
the converse implication does not hold). Due to the presence of infinitely many
variables in $`l$-calculus, for each term $T$ with bound variables there are
countably many inhabitants in ${[T]}_{\alpha}$. We are therefore interested in
counting $\alpha$-equivalence classes rather than particular $`l$-terms. In order to
deal with the issue of $\alpha$-equivalence, we consider $`l$-terms in the unary
de~Bruijn notation (see, e.g.~\cite{gry_les}) in which $`l$-terms
are canonical representatives of $`a$-equivalence classes. For that reason, we are in 
fact counting $\alpha$-equivalence classes of regular $`l$-terms.

Consider the following natural way of defining the size of $`l$-terms, in which all the constructors contribute one to the overall term size. This means that abstractions, applications, successors and zeros are all of size one. Formally,
\begin{eqnarray*}
  |`l N| &=& |N| +1,\\
  |N\,M|&=& |N| + |M| +1,\\
  |S \underbar{\sf n}| &=& |\underbar{\sf n}| +1,\\
  |\zerodot| &=&1.
\end{eqnarray*}
For instance, the $`l$-term for \textbf{K} which is traditionally written as $`l x . `l y. x$, in the de~Bruijn model is written as $`l `l S \zerodot$. We
have $| `l `l S \zerodot| = 4$ as there are two $`l$-abstractions, one successor $S$ and one $\zerodot$.  The
$`l$-term for \textbf{S} (which should not be confused with the successor symbol) is written as $`l x.`ly.`l z. x z (y z)$, whereas  using de~Bruijn indices we write $`l`l`l (((S S \zerodot) \zerodot) ((S \zerodot) \zerodot))$. Its size is equal to $13$ since there are three $`l$-abstractions, three applications, three successors $S$'s, and four $\zerodot$'s. 

\subsection{Plain $\lambda$-terms}\label{subsec:plain-lambda-terms}
In this section we are interested in the generating function for the
    sequence corresponding to the numbers of $`l$-terms. Let us start by 
    considering the class of de Bruijn indices.

    \begin{prop}\label{prop:D-gen-fun}
	Let $D(z)$ stand for the generating function enumerating de~Bruijn indices. Then
            \[ D(z) = \frac{z}{1-z} = \sum_{n=1}^{\infty} z^n.\]
    \end{prop}

    \begin{proof}
        Let $n \in \nats$. There exists a unique de~Bruijn
        index $\underbar{\sf n}$ encoding $n$. Since application and $\zerodot$ are
        both of size $1$, the size of $\underbar{\sf n}$ is equal to $n + 1$ and
        thus ${\seq{\nthcoeff{D(z)}}}_{n\in\nats} = \seq{0,1,1,\ldots}$, which
        immediately implies $D(z) = \frac{z}{1-z}$. 
    \end{proof}
    
    \begin{prop}\label{prop:Loo-gen-fun}
    Let $\gL$ stand for the generating function enumerating all $`l$-terms. Then
        \[ \gL = \frac{{(1-z)}^{3/2} - \sqrt{1-3z-z^2-z^3}}{2z\sqrt{1-z}}. \]
    \end{prop}

    \begin{proof}
        Since $`l$-terms are either applications, abstractions or de~Bruijn
        indices, the set $\cgL$ of lambda terms can be expressed as
        \[ \cgL  = \cgL \cgL `(+) `l \cgL `(+) \D. \]   
        Using this representation, we immediately obtain a corresponding quadratic
        equation defining the generating function $\gL$:
        \begin{equation}\label{eq:Loo-definition}
            \gL = z \gL^2 + z \gL + \frac{z}{1-z}.
        \end{equation}
        Computing its discriminant $\DeltagL = 
        \frac{1 - 3z - z^2 - z^3}{1-z}$ we finally solve the above equation:
        \begin{eqnarray*}
          \gL &=& \frac{(1-z) - \sqrt{\DeltagL}}{2z}\\
              &=& \frac{{(1-z)}^{3/2} - \sqrt{1-3z-z^2-z^3}}{2z\sqrt{1-z}}.
        \end{eqnarray*}
    \end{proof}
    
    Using the generating function $\gL$ we can now easily find the asymptotic
    growth rate of the sequence ${\seq{\nthcoeff{\gL}}}_{n\in\nats}$.
    \begin{theorem}\label{th:Loo-approx}
	The asymptotic approximation of the number of $`l$-terms of size $n$ 
	is given by
            \[ \nthcoeff{\gL} \sim {\left(3.38298\ldots \right)}^n
            \frac{C}{n^{\nicefrac{3}{2}}}, \quad \text{where} \quad C \doteq 0.60676. \]
    \end{theorem}
    
    \begin{proof}
    	    Examining the function $\gL$ we note that its dominating
        singularity $\domsing{L_\infty}$ is equal to the root of 
        smallest modulus of $1-3z-z^2-z^3$. Therefore,
        \[\domsing{L_\infty} = \frac{1}{3}
        \left(\sqrt[3]{26+6 \sqrt{33}} -\frac{4\ 2^{2/3}}{\sqrt[3]{13+3
        \sqrt{33}}} - 1\right) \doteq 0.29559774252208393\]
        and hence $\nicefrac{1}{\domsing{L_\infty}} \doteq 3.38298$.
        Let us write $\gL$ as
        \begin{eqnarray*}
                \gL &=& \frac{(1-z)-\sqrt{\frac{1-3z-z^2-z^3}{1-z}}}{2z}\\
                    &=& \frac{(1-z)-\sqrt{\frac{(\domsing{L_\infty} -z)\cdot Q(z)}{1-z}}}{2z},
        \end{eqnarray*}
        for the appropriate polynomial $Q(z)$.
        Applying~\autoref{th:standard-func-scale} and~\autoref{th:newton-puiseux} we obtain
        \[\nthcoeff{\gL} \sim
                {\left(\frac{1}{\domsing{L_\infty}}\right)}^n \cdot
            \frac{n^{-3/2}}{`G(-\frac{1}{2})}~\widetilde{C} \qquad
            \text{with}
        \qquad \widetilde{C} = \frac{-\sqrt{\domsing{L_\infty} \frac{Q(\domsing{L_\infty})}
        {1-\domsing{L_\infty}}}}{2\domsing{L_\infty}}.\]
        Since $Q(\domsing{L_\infty}) \doteq 
        								3.85321718036529$, we finally get
        \[C = \frac{\widetilde{C}}{`G(-\frac{1}{2})} \doteq 0.60676.\]
    \end{proof}
    
    The sequence ${\seq{\nthcoeff{\gL}}}_{n\in\nats}$ is known as
    \href{http://oeis.org/A105633}{\textbf{A105633}} in the
    \emph{Online Encyclopedia of Integer Sequences}~\cite{oeisEncyclopedia} and counts, beside plain $`l$-terms, \emph{black-white binary trees} (described in Section~\ref{subsec:bw-trees}) and \emph{binary trees without zigzags} (described in Section~\ref{subsec:zz-trees}).
    Its first $15$ values are as follows:
    \begin{center}
    \medskip
    \begin{small}
       0,~~1,~~2,~~4,~~9,~~22,~~57,~~154,~~429,~~1223,~~3550,~~10455,~~31160,~~93802,
        ~~284789.
    \end{small}
    \medskip
    \end{center}
%
    
    \subsection{Holonomic representation of $\gL$}\label{subsec:linfinity-holonomic}

    Using the \textsf{Maple} package \textsf{gfun} (see~\cite{SalvyZimmermann1994}) we
    find the following holonomic equation defining $\gL$:
    \[ z^3 + z^2 - 2z + (z^3 + 3 z^2 -3z +1)
        \gL + (z^5 + 2z^3 -4 z^2 + z) L_{\infty}'(z) = 0 \]
    	with $L_\infty(0) = 0$.
    Such an implicit form of $\gL$ allows us to derive a simpler, compared to the
    combinatorial definition, recursive definition of
    its coefficients. For convenience, let us denote $L_{\infty,n} :=
    \nthcoeff{\gL}$. Now, we can express the recursive definition of
    $L_{\infty,n}$ as:

    \begin{align*}
            L_{\infty,0} &= 0,\qquad L_{\infty,1} = 1,\qquad
            L_{\infty,2} = 2,\qquad L_{\infty,3} = 4,\\
            L_{\infty,n} &= \frac{(4n-1)L_{\infty,n-1} - (2n-1)L_{\infty,n-2} -
    L_{\infty,n-3}- (n-4)L_{\infty,n-4}}{n+1}.
    \end{align*}
  
    Note that $L_{\infty,n}$ depends on the previous four values
    $L_{\infty,n-1}$, $L_{\infty,n-2}$, $L_{\infty,n-3}$ and $L_{\infty,n-4}$.
    Exploiting this fact, the above definition allows us to compute the exact
    value $L_{\infty,n}$ using only linear number of arithmetic operations.
    Moreover, we note that this holonomic equation could be used to develop a
    random generator in the spirit of \cite{BBJ}.
    
   \subsection{$E$-free black-white binary trees}\label{subsec:bw-trees}
   
   A \emph{black-white binary tree} is a binary tree in which nodes are coloured
    either \emph{black} $\bullet$ or \emph{white} $`(?)$. Let $E$ be a set of
    edges. An \emph{$E$-free black-white binary tree} is a black-white binary
    tree in which edges from the set $E$ are forbidden. For instance, if the set of
    forbidden edges is $E_1 = \{\lf{`(!)}{`(?)}, \ri{`(!)}{`(?)},
    \ri{`(!)}{`(!)}, \ri{`(?)}{`(?)}\}$, then the only allowed edges are $A_1 =
    \{\lf{`(?)}{`(!)}, \lf{`(!)}{`(!)}, \lf{`(?)}{`(?)}, \ri{`(?)}{`(!)}\}$.
    The \emph{size} of a black-white tree is the total number of its nodes.
    For $E_1$, like for $E_2 = \{\ri{`(?)}{`(!)}, 
    \lf{`(?)}{`(!)}, \lf{`(!)}{`(!)}, \lf{`(?)}{`(?)} \}$, which is obtained by left-right
    symmetry, the $E$-free black-white binary trees are counted by
    \href{http://oeis.org/A105633}{\textbf{A105633}}, see~\cite{DBLP:journals/dm/GuLM08}. Henceforth we consider only the set $E_1$ and speak rather in terms of
    allowed edge patterns, i.e.~$A_1$. For convenience, whenever we use
    the term \emph{black-white trees}, we mean the black-white trees with
    allowed set of patterns $A_1$. Unless otherwise stated, we assume that
    black-white trees have black roots.
    
    \medskip
    Let $\B\W_{`(!)}$ and $\B\W_{`(?)}$ denote the set of black-white trees with
    a black, respectively white, root. Interpreting the set of allowed edges $A_1$
    combinatorially, we can define both $\B\W_{`(!)}$ and $\B\W_{`(?)}$ using the
    following mutually recursive equations:
    \begin{eqnarray*}
            \B\W_{`(!)}\, &=&\, `(!)\, `(+)\, 
            \raisebox{.2cm}{\xymatrix@=.2pc{&{`(!)}\ar@{-}[dl]\\
            \B\W_{`(!)}}} \, `(+)\, \raisebox{.2cm}{\xymatrix@=.2pc{&{`(!)}
            \ar@{-}[dl]\\ \B\W_{`(?)}}}\\
            \B\W_{`(?)}\, &=& \, `(?)\, `(+)\,  
            \raisebox{.4cm}{\xymatrix@=.2pc{&{`(?)}\ar@{-}[dl]\\ \B\W_{`(?)}}} \ `(+)
            \raisebox{.4cm}{\xymatrix@=.2pc{&{`(?)}\ar@{-}[dr]\\
            &&\B\W_{`(!)}}}\, `(+)\, 
            \raisebox{.4cm}{\xymatrix@=.2pc{&{`(?)}\ar@{-}[dl]\ar@{-}[dr]\\
            \B\W_{`(?)}&&\B\W_{`(!)}}} 
    \end{eqnarray*}

    Such a representation yields the following identities on the corresponding
    generating functions $BW_{`(!)}(z)$ and $BW_{`(?)}(z)$:
    \begin{eqnarray*}
        BW_{`(!)}(z) &=& z + z BW_{`(!)}(z) + z BW_{`(?)}(z), \\
        BW_{`(?)}(z) &=& z + z BW_{`(?)}(z) + z BW_{`(!)}(z) + 
            z BW_{`(?)}(z)BW_{`(!)}(z).
    \end{eqnarray*}
    Reformulating this system, we obtain
    \[ BW_{`(?)}(z) = \frac{(1-z) BW_{`(!)}(z) - z}{z}, \]
    and hence
    \[ (1-z)z BW_{`(!)}^{2}(z) - {(1-z)}^{2} BW_{`(!)}(z) + z = 0.\]
    
    Notice that the equation defining $BW_{`(!)}(z)$ is equivalent to the equation
    (\ref{eq:Loo-definition}) defining $\gL$ up to multiplication by $(1-z)$. It
    follows that both ${\left(\nthcoeff{BW_{`(!)}(z)}\right)}_{n\in\nats}$ and
    ${\left(\nthcoeff{\gL}\right)}_{n\in\nats}$ are equal and therefore there
    exists a bijection between $`l$-terms and black-white trees.
    
    \subsection{Bijection between $`l$-terms and black-white trees}\label{subsec:bijection-black-white}
    We are now ready to give a bijective translation $\ltobw$ from $`l$-terms to
    black-white trees and the inverse translation $\bwtol$ from black-white
    trees to $`l$-terms:

    \begin{align*}
            \zerodot \, &\xrightarrow{\ltobw}\,  `(!) & `(!) \,
            &\xrightarrow{\bwtol}\, \zerodot\\ S\, n \, &\xrightarrow{\ltobw}\,
            \lf{`(!)}{\ltobw(n)} & \lf{`(!)}{t} \, &\xrightarrow{\bwtol}\, S\,
            \bwtol(t)\\ `l\, M \, &\xrightarrow{\ltobw}\, \lf{`(?)}{\ltobw(M)} &
            \lf{`(?)}{t} \, &\xrightarrow{\bwtol}\, `l\, \bwtol(t)\\ M_1\,M_2\,
                            &\xrightarrow{\ltobw}\, \raisebox{.5cm}{${\tiny
            \xymatrix@=.2pc{&{\ltobw(M_2)}\ar@{-}[dl]\\{`(?)}\ar@{-}[dr]\\
                            &{\ltobw(M_1)}}}$} & \raisebox{.5cm}{${\tiny
                            \xymatrix@=.2pc{&{t_2}\ar@{-}[dl]\\{`(?)}\ar@{-}[dr]\\
                                            &{t_1}}}$}\, &\xrightarrow{\bwtol}\,
            \bwtol(t_1)\, \bwtol(t_2)
    \end{align*}

    \begin{prop}
      Both $\ltobw$ and $\bwtol$ are mutually inverse bijections,
      i.e.\[\bwtol \circ \ltobw = id_{`L} \qquad \text{and} \qquad \ltobw \circ
      \bwtol = id_{\B\W_{`(!)}}.\]
    \end{prop}
   
    In order to translate a given black-white tree $t$ into a corresponding
    $`l$-term, we decompose $t$ depending on the type of its leftmost node.
    If $t$ is a single black node $\bullet$, we translate it into
    $\zerodot$. Otherwise, we have to consider three cases based on the set
    $A_1$ of allowed edges and map them into $`l$-abstraction, successor, or 
    application, respectively.
    
    \begin{example}
    Let us give two black-white trees corresponding to:
    \begin{itemize}
        \item $\mathsf{`W} = (`l.x x)(`l.x x) =
            (`l(\zerodot\,\zerodot))\,`l(\zerodot\,\zerodot)$, and
        \item $\mathsf{Y} = `lf.(`lx.f(xx))(`lx.f(xx)) =
            `l(`l(S\,\zerodot\,(\zerodot\,\zerodot))\,`l( S\,
            \zerodot\,(\zerodot\,\zerodot)))$ 
    \end{itemize}
   
    \begin{displaymath}
        \begin{array}{l@{\qquad\qquad}l@{\qquad\qquad}l}
            \mathsf{\ltobw(`W)}  %
            \xymatrix@=.2pc{
                &&&&`(!) \ar@{-}[dl]\\
                &&&`(?) \ar@{-}[dl] \ar@{-}[dr]\\
                &&`(?) \ar@{-}[dl]&&`(!)\\
                &`(?) \ar@{-}[dr]\\
                &&`(!) \ar@{-}[dl]\\
                & `(?) \ar@{-}[dl] \ar@{-}[dr]\\
                `(?)&&`(!)
            }
        &
            \mathsf{\ltobw(Y)}  %
            \xymatrix@=.2pc{
                &&&&&&&`(!)\ar@{-}[dl]\\
                &&&&&&`(?)\ar@{-}[dl]\ar@{-}[dr]\\
                &&&&&`(?)\ar@{-}[dl]\ar@{-}[dr]&&`(!)&&\\
                &&&&{`(?)}\ar@{-}[dl] && `(!)\ar@{-}[dl]\\
                &&&`(?) \ar@{-}[dl]\ar@{-}[dr]&& `(!)\\
                &&`(?)&&`(!)\ar@{-}[dl]\\
                &&&`(?)\ar@{-}[dl]\ar@{-}[dr]\\
                &&`(?)\ar@{-}[dl]\ar@{-}[dr]&&`(!)&&\\
                &{`(?)} && `(!)\ar@{-}[dl]\\
                && `(!)
            }
        \end{array}
    \end{displaymath}
    \end{example}
    
    We provide Haskell implementations of $\ltobw$ and $\bwtol$ which can be
    found at~\cite{mb-haskell-implementation}. Our implementations were
    tested using Quickcheck~\cite{Claessen-2000}.

\subsection{Binary trees without zigzags}\label{subsec:zz-trees}
In this section we are interested in zigzag-free binary trees, i.e.~trees without
a~forbidden \emph{zigzag} subtree:
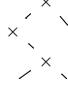
\begin{displaymath}
  \tiny \xymatrix @=.05pc{&{`*}\ar@{-}[dl]\ar@{-}[dr]\\
    `*\ar@{-}[dr]&&\\%
    &{`*}\ar@{-}[dl]\ar@{-}[dr] \\ &&&&}
\end{displaymath}

Let us denote $\B\Z_1$ the set of zigzag free trees. The above negative criterion
can be stated positively.  Wherever inside such a tree we start from a node by a left
branch and follow only left branches, we get to an isolated node $`*$, i.e.~a leaf.
This description can be translated into the following combinatorial equations:
\begin{eqnarray*}
  \B\Z_1\, &=& \, \raisebox{.4cm}{\xymatrix
               @=.05pc{{`*}\ar@{-}[dr]\\&\B\Z_1}} \ `(+)\  \B\Z_2\\
  \B\Z_2\, &=&\, {`*} \ \ `(+)\ \  \raisebox{.4cm}{\xymatrix
               @=.05pc{&{`*}\ar@{-}[dl]\\ \B\Z_2}} \ `(+)\
  \raisebox{.4cm}{\xymatrix
  @=.05pc{&{`*}\ar@{-}[dl]\ar@{-}[dr]\\\B\Z_2&&\B\Z_1}}  
\end{eqnarray*}
    Similarly to $\gL$ and $BW_{`(!)}(z)$, the generating function $BZ_1(z)$ can be
    expressed as a solution of the functional equation:
    \[ z(1-z) BZ_1^2(z) + {(1-z)}^2 BZ_1(z) + z = 0.\] 
    
    It follows that the sequence ${\left(\nthcoeff{BZ_1(z)}\right)}_{n\in\nats}$ is
    equal to ${\left(\nthcoeff{BW_{`(!)}(z)}\right)}_{n\in\nats}$ and also to
    ${\left(\nthcoeff{\gL}\right)}_{n\in\nats}$, suggesting that appropriate
    bijections exist. We note that Sapounakis et
    al.~\cite{DBLP:journals/dm/SapounakisTT06} consider the same sequence
    defined in terms of constrained Dyck paths and give the following explicit
    formula: 
    \[\nthcoeff{BZ_1(z)} = \nthcoeff{\gL} = \sum_{k=0}^{(n-1) \div 2}
    \frac{{(-1)}^k}{n-k}\binom{n-k}{k}\binom{2n-3k}{n-2k-1}. \]

    \subsection{Bijection between black-white trees and zigzag-free trees}\label{subsec:bijection-zigzag}
    We start by giving a bijective translation $\bwtobz$ from black-white trees to
    zigzag-free ones. For convenience, we use $u_1$ and $u_2$ to denote
    arbitrary (possibly empty) black-white trees.
    
    \begin{align*}
            `(!) \, &\xrightarrow{\bwtobz}\, `* &\\
            \raisebox{.4cm}
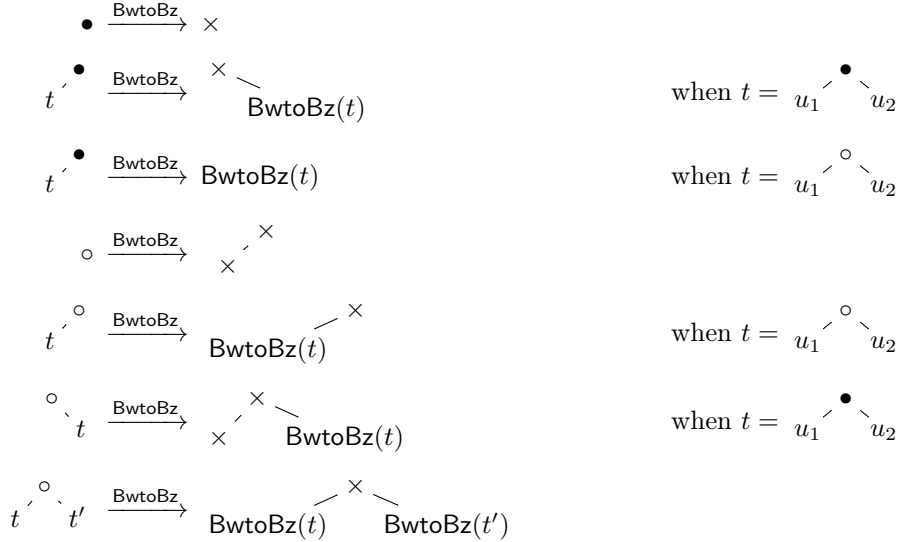
{\xymatrix@=.05pc{&{`(!)}\ar@{-}[dl]\\ t}}\, 
            &\xrightarrow{\bwtobz}\, \raisebox{.4cm}{\xymatrix
    @=.05pc{{`*}\ar@{-}[dr]\\ &{\scriptsize \bwtobz}(t)}}&
            \text{when~} t = \raisebox{.4cm}{\xymatrix@=.05pc{&{`(!)}
            \ar@{-}[dl]\ar@{-}[dr]\\u_1&& u_2}}\\
            \raisebox{.4cm}{\xymatrix@=.05pc{&{`(!)}\ar@{-}[dl]\\ t}}\,
                                             &\xrightarrow{\bwtobz}\,
                    \bwtobz(t)\, &  \text{when~} t = \raisebox{.4cm}
            {\xymatrix@=.05pc{&{`(?)}\ar@{-}[dl]\ar@{-}[dr]\\u_1&& u_2}}\\
            `(?) \, &\xrightarrow{\bwtobz}\, \raisebox{.4cm}
            {\xymatrix@=.05pc{&{`*}\ar@{-}[dl]\\~{`*}&&}}&\\
            \raisebox{.4cm}{\xymatrix@=.05pc{&{`(?)}\ar@{-}[dl]\\ t}}\, 
            &\xrightarrow{\bwtobz}\, \raisebox{.4cm}
            {\xymatrix@=.05pc{&{`*}\ar@{-}[dl]\\{\scriptsize \bwtobz}(t)&&}}&
            \text{when~} t = \raisebox{.4cm}{\xymatrix@=.05pc{&{`(?)}
            \ar@{-}[dl]\ar@{-}[dr]\\u_1&& u_2}}\\
            \raisebox{.4cm}{\xymatrix@=.05pc{{`(?)}\ar@{-}[dr]\\ 
            &t}}\, &\xrightarrow{\bwtobz}\, 
            \raisebox{.4cm}{\xymatrix@=.05pc{&{`*}\ar@{-}[dl]\ar@{-}[dr]\\
            `*&&{\scriptsize \bwtobz}(t)}} & \text{when~} t = \raisebox{.4cm}
            {\xymatrix@=.05pc{&{`(!)}
            \ar@{-}[dl]\ar@{-}[dr]\\u_1&& u_2}}\\
            \raisebox{.4cm}{\xymatrix@=.05pc{&{`(?)}\ar@{-}[dl]\ar@{-}[dr]\\ t&&t'}} 
            \, &\xrightarrow{\bwtobz}\, \raisebox{.4cm}{\xymatrix@=.05pc{&{`*}\ar@{-}[dl]
            \ar@{-}[dr]\\ {\scriptsize \bwtobz}(t)&&{\scriptsize \bwtobz}(t')}}%
                    &\\
    \end{align*}
    
    \begin{prop}
	Let $t$ be a black-white tree. Then trees $t$ and $\bwtobz(t)$ are of equal size.
	\end{prop}    
	
	\begin{proof}
	Let us notice that it suffices to consider the case $\bwtobz \big( \large \lf{t}{`(!)} \big)$, 
	since it results in subtracting one black node. Because the
	root of $t$ is white, the next translation step is done according to
	one of the last four rules, which eventually falls into either the fourth 
	or the sixth equation. Since both of them enforce adding one additional 
	$`*$, the total number of nodes is preserved.
	\end{proof}  

	What remains is to give the inverse translation,
    which we present as two mutually recursive functions $\bztobw_{`(!)}$ and
    $\bztobw_{`(?)}$:
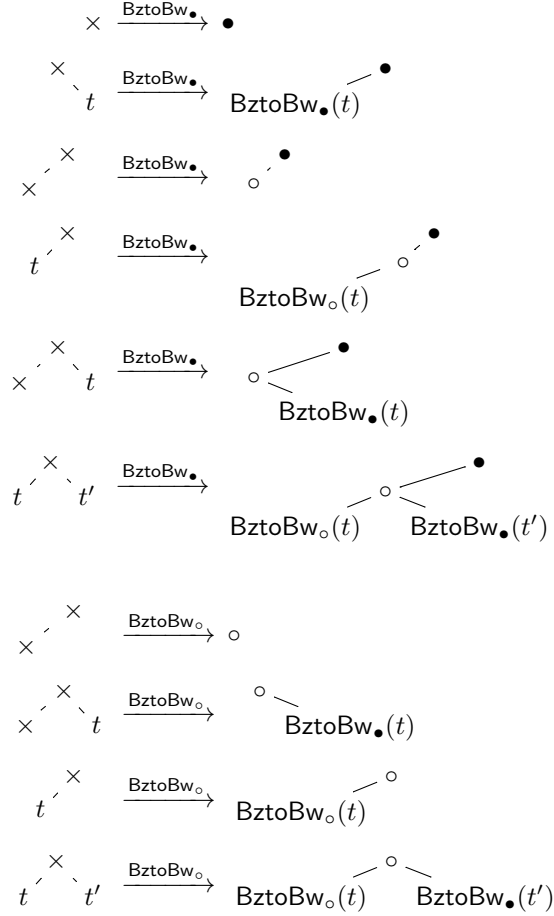
    \begin{align*}
            `* \, &\xrightarrow{\bztobw_{`(!)}}\, `(!) &\\
            \raisebox{.4cm}{\xymatrix@=.05pc{{`*}\ar@{-}[dr]\\&t}}
            \, &\xrightarrow{\bztobw_{`(!)}}\, \raisebox{.4cm}{\xymatrix@=.05pc{&{`(!)}
            \ar@{-}[dl]\\\bztobw_{`(!)}(t)}}&\\
            \raisebox{.4cm}{\xymatrix@=.05pc{&{`*}\ar@{-}[dl]\\~{`*}&&}}
            \, &\xrightarrow{\bztobw_{`(!)}}\, \raisebox{.4cm}{\xymatrix@=.05pc{&&{`(!)}
            \ar@{-}[dl]\\&`(?)}}&\\
            \raisebox{.4cm}{\xymatrix@=.05pc{&{`*}\ar@{-}[dl]\\t&&}}
            \, &\xrightarrow{\bztobw_{`(!)}}\, \raisebox{.4cm}{\xymatrix@=.05pc{&&{`(!)}
            \ar@{-}[dl]\\&~`(?)\ar@{-}[dl]\\~\bztobw_{`(?)}(t)&&
            }}&\\
            \raisebox{.4cm}{\xymatrix@=.05pc{&{`*}\ar@{-}[dl]\ar@{-}[dr]\\{`*}&&t}}
            \, &\xrightarrow{\bztobw_{`(!)}}\, \raisebox{.4cm}{\xymatrix@=.05pc{&&{`(!)}
            \ar@{-}[dl]\\&`(?)\ar@{-}[dr]\\&&
            \bztobw_{`(!)}(t)}}&\\
            \raisebox{.4cm}{\xymatrix@=.05pc{&{`*}\ar@{-}[dl]\ar@{-}[dr]\\t&&t'}}
            \, &\xrightarrow{\bztobw_{`(!)}}\, \raisebox{.4cm}{\xymatrix@=.05pc{&&{`(!)}
            \ar@{-}[dl]\\&`(?)\ar@{-}[dl]\ar@{-}[dr]\\\bztobw_{`(?)}(t)&&
            \bztobw_{`(!)}(t')}}&
    \end{align*}
        
    \begin{align*}
            \raisebox{.4cm}{\xymatrix@=.05pc{&~{`*}\ar@{-}[dl]\\`*&&}}
            \, &\xrightarrow{\bztobw_{`(?)}}\, `(?)&\\
            \raisebox{.4cm}{\xymatrix@=.05pc{&{`*}\ar@{-}[dl]\ar@{-}[dr]\\~`*&&t}}
            \, &\xrightarrow{\bztobw_{`(?)}}\,
            \raisebox{.4cm}{\xymatrix@=.05pc{&{`(?)}\ar@{-}[dr]\\
    &&\bztobw_{`(!)}(t)}} &\\
            \raisebox{.4cm}{\xymatrix@=.05pc{&{`*}\ar@{-}[dl]\\t&&}}
            \, &\xrightarrow{\bztobw_{`(?)}}\,
            \raisebox{.4cm}{\xymatrix@=.05pc{&{`(?)}\ar@{-}[dl]\\
    \bztobw_{`(?)}(t)&&}} &\\
            \raisebox{.4cm}{\xymatrix@=.05pc{&{`*}\ar@{-}[dl]\ar@{-}[dr]\\t&&t'}}
            \, &\xrightarrow{\bztobw_{`(?)}}\,
            \raisebox{.4cm}{\xymatrix@=.05pc{&{`(?)}\ar@{-}[dl]\ar@{-}[dr]\\
    \bztobw_{`(?)}(t)&&\bztobw_{`(!)}(t')}} &
    \end{align*}
    
    \begin{prop}
	Let $t$ be a zigzag-free tree. Then trees $t$ and $\bztobw_{`(!)}(t)$ 
	are of equal size.
	\end{prop}

	\begin{proof}
    The fourth and the sixth equations defining $\bztobw_{`(!)}$ introduce an
    additional white node $`(?)$, but since both the first and the second equations
    of $\bztobw_{`(?)}$ remove one node, the overall tree size is preserved.
	\end{proof}
	
	\begin{prop}
            Both $\bztobw_{`(!)}$ and $\bwtobz$ are mutually inverse bijections, i.e.
            \[\bztobw_{`(!)} \circ \bwtobz = id_{\B\W_{`(!)}} \qquad \text{and}
                            \qquad \bwtobz \circ
                    \bztobw_{`(!)} = id_{\B\Z}.\]
    \end{prop}

    \begin{example}
    Let us present the zigzag-free tree corresponding to the aforementioned
    black-white tree associated with $\mathsf{`W}$:
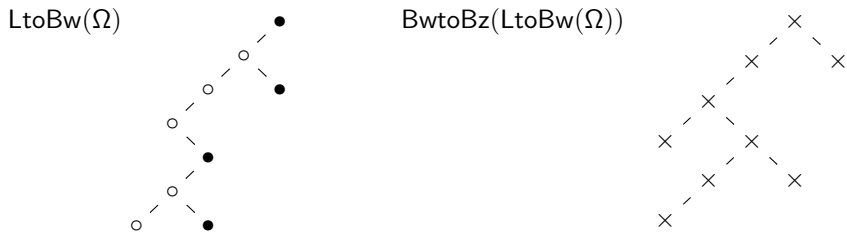
    \begin{displaymath}
        \begin{array}{l@{\qquad\qquad}l@{\qquad\qquad}l}
            \mathsf{\ltobw(`W)}  %
            \xymatrix@=.2pc{
                &&&&`(!) \ar@{-}[dl]\\
                &&&`(?) \ar@{-}[dl] \ar@{-}[dr]\\
                &&`(?) \ar@{-}[dl]&&`(!)\\
                &`(?) \ar@{-}[dr]\\
                &&`(!) \ar@{-}[dl]\\
                & `(?) \ar@{-}[dl] \ar@{-}[dr]\\
                `(?)&&`(!)
            }
        &
            \mathsf{\bwtobz(\mathsf{\ltobw(`W)})}  %
            \xymatrix@=.2pc{
                &&&&`*\ar@{-}[dl]\ar@{-}[dr]&\\
                &&&`*\ar@{-}[dl]&&`*\\
                &&`*\ar@{-}[dl]\ar@{-}[dr]\\
                &`*&&`*\ar@{-}[dl]\ar@{-}[dr]\\
                &&`*\ar@{-}[dl]&&`*\\
                &`*\\
            }
        \end{array}
    \end{displaymath}
    \end{example}
    
    We provide Haskell implementations of $\bwtobz$, $\bztobw_{`(!)}$ and 
    $\bztobw_{`(?)}$ which can be found at~\cite{mb-haskell-implementation}.
    Our implementations were tested using Quickcheck~\cite{Claessen-2000}.

\subsection{Neutral $`l$-terms and $\beta$-normal forms}\label{subsec:nf-lambda-terms}
Here we are interested in the class $\N$ of $\beta$-normal forms, i.e.~$`l$-terms which do not have subterms of the form $(`l N)\,M$, and the associated class $\M$ of neutral terms, i.e.~normal forms without head abstractions. We start by giving a combinatorial specification of normal forms involving the class $\M$ of neutral terms:
\begin{eqnarray*}
  \N &=& \M `(+) `l\, \N\\
  \M &=& \M \N `(+) \D  \\
  \D &=& S\, \D `(+) \zerodot
\end{eqnarray*}

Normal forms either are neutral or start with a head abstraction. Neutral terms, in turn, are either de~Bruijn indices, or are in form of an application of a neutral term to a normal form. The above specification yields the following system of equations
for the corresponding generating functions:
\begin{eqnarray*}
  N(z) &=& M(z) + z N(z),\\
  M(z) &=& z M(z) N(z) + D(z),\\
  D(z) &=& z D(z) + z.
\end{eqnarray*}
Solving this system, we obtain the following generating functions:
\begin{eqnarray*}
  M(z) &=& \frac{1-z-\sqrt{(1+z)(1-3z)}}{2z},\\
  N(z) &=& \frac{M(z)}{1-z}.
\end{eqnarray*}

Note that $M(z)$ is the generating function corresponding to the counting sequence of Motzkin numbers (see, e.g.~\cite[p.~396]{Flajolet:2009:AC:1506267}), for convenience denoted henceforth as $\T$. Naturally, it means that there exists a size-preserving bijection between Motzkin trees and neutral forms.

\subsection{Bijection between Motzkin trees and neutral forms}\label{subsec:bijection-motzkin-trees}
Let $u_n$ denote the unary Motzkin path of size $n > 0$. We start by defining two auxiliary operations \textsf{UnToL} and \textsf{UnToD}, translating unary Motzkin paths into $\lambda$-paths and de~Bruijn indices, respectively:
    
    \begin{align*}
    `(!) \, &\xrightarrow{\untol}\, `l & `(!) \, &\xrightarrow{\untod}\, \zerodot\\
    \raisebox{.4cm}{\upath{\bullet}{u_n}} \, &\xrightarrow{\untol}\, \raisebox{.6cm}{\upath{\lambda}{\untol\left(u_n\right)}} & \raisebox{.4cm}{\upath{\bullet}{u_n}} \, &\xrightarrow{\untod}\, \raisebox{.6cm}{\upath{S}{\untod\left(u_n\right)}}\\
    \end{align*}

Using \textsf{UnToL} and \textsf{UnToD} we can now define a bijective translation \textsf{MoToNe} from Motzkin trees to corresponding neutral terms:

\begin{align*}
	u_n \, &\xrightarrow{\motone}\, \untod\left(u_n\right)\\
    \raisebox{.9cm}{\bcross{u_n}{t}{t'}} \, &\xrightarrow{\motone}\, \raisebox{.9cm}{\btcross{@}{\motone\left(t\right)}{\untol\left(u_n\right)}{\motone\left(t'\right)}}\\
    \raisebox{.5cm}{\bnode{\bullet}{t}{t'}} \, &\xrightarrow{\motone}\, \raisebox{.6cm}{\bnode{@}{\motone\left(t\right)}{\motone\left(t'\right)}}
    \end{align*}

\begin{prop}
    $\motone$ is a bijection.
\end{prop}

\begin{proof}
    The proposition is an easy consequence of the fact that
    $\motone$ preserves the exact number of unary and binary nodes.
\end{proof}

In order to translate Motzkin trees to corresponding neutral terms we have to 
consider two cases. Either we are given a Motzkin tree starting with
a unary node or a Motzkin tree starting with a binary node. The second case is straightforward due to the fact
that binary nodes correspond to neutral term applications. Assume we are given a
Motzkin tree starting with a unary path $u_n$ of size $n$. We have to decide
whether the path corresponds to a de~Bruijn index or to a chain of
$`l$-abstractions. This distinction is uniquely determined by the existence of
the path's \emph{splitting node} -- the binary node directly below $u_n$. If
$u_n$ has a splitting node, then it corresponds to a chain of $n$
$`l$-abstractions which will be placed on top of the corresponding right neutral
term constructed recursively from $u_n$'s splitting node. Otherwise, $u_n$
corresponds to the $n$th de~Bruijn index.

\medskip
What remains is to give the inverse translation \textsf{NeToMo} from neutral terms to Motzkin trees. Let \textsf{LToUn} and \textsf{DToUn}
denote the inverse functions of \textsf{UnToL} and \textsf{UnToD}, respectively. Let
$l_n$ denote the unary $\lambda$-path of size $n > 0$. The translation \textsf{NeToMo} is given by:

\begin{align*}
	\underline{n} \, &\xrightarrow{\netomo}\, \dtoun\left(\underline{n}\right)\\
    \raisebox{.9cm}{\btcross{@}{t}{l_n}{t'}} \, &\xrightarrow{\netomo}\raisebox{.9cm}{\bcross{\ltoun\left(l_n\right)}{\netomo\left(t\right)}{\netomo\left(t'\right)}}\\
    & \qquad \qquad \qquad \text{where } t' \text{ does not start with a head } \lambda\\
    \raisebox{.6cm}{\bnode{@}{t}{t'}} \, &\xrightarrow{\netomo}\, \raisebox{.5cm}{\bnode{\bullet}{\netomo\left(t\right)}{\netomo\left(t'\right)}}
    \end{align*}

\begin{prop}
  $\motone `(?)\netomo = id_{\M}$ and $\netomo `(?)\motone = id_{\T}$.
\end{prop}

\begin{example}
	Consider the neutral term $P = \zerodot \left(`l `l \zerodot \left(S\, \zerodot\right)\right)$. The following figure presents $P$ and its Motzkin tree counterpart through the translation $\motone$.
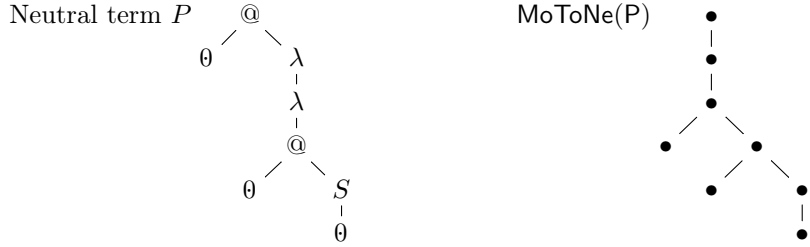
    \begin{displaymath}
        \begin{array}{l@{\qquad\qquad}l@{\qquad\qquad}l}
            \text{Neutral term } P  %
            \xymatrix@=.3pc{
            	   & @ \ar@{-}[dl] \ar@{-}[dr] & & &\\
            	   \zerodot & & `l \ar@{-}[d] & &\\
            	   & & `l \ar@{-}[d] & &\\
            	   & & @ \ar@{-}[dl] \ar@{-}[dr] & &\\
            	   & \zerodot & & S \ar@{-}[d] & &\\
            	   & & & \zerodot &
            }
        &
            \mathsf{\motone(\mathsf{P})}  %
            \xymatrix@=.5pc{
            	   & `(!) \ar@{-}[d] & & &\\
            	   & `(!) \ar@{-}[d] & & &\\
            	   & `(!) \ar@{-}[dl] \ar@{-}[dr] & & &\\
            	   `(!) & & `(!) \ar@{-}[dl] \ar@{-}[dr] & &\\
            	   & `(!) & & `(!) \ar@{-}[d] &\\
            	   & & & `(!) &
            }
        \end{array}
    \end{displaymath}
    \end{example}
    
Let us notice that the simple translation $\netomo$ allows us to design an effective exact-size sampler for neutral $`l$-terms in the natural size notion, employing the sampler for Motzkin trees of Bacher et al.~\cite{BBJ}. Given a number $n \in \nats$, we sample a uniformly random Motzkin tree of size $n$, constructing a corresponding neutral $`l$-term out of it using the $\netomo$ translation. The resulting outcome is clearly a uniformly random neutral $`l$-term of size $n$. As our translation is linear in time and space, the overall complexity of the described sampler is, on average, linear in both time and space.

\subsection{Head normal forms}\label{subsec:head-normal-forms}
In this section we are interested in counting head normal forms, i.e.~$`l$-terms without head redexes and the associated auxiliary set $\K$ of neutral head normal forms, as defined by the following combinatorial specification:
\begin{eqnarray*}
  \calH &=& \K `(+) `l \calH\\
  \K &=& \K \cgL `(+) \D
\end{eqnarray*}

A head normal form either starts with a head $`l$-abstraction followed by another head normal form, or is a neutral head normal form. In the latter case, it must be a de~Bruijn index or an application of a neutral head normal form to an arbitrary $`l$-term. Translating the above specification into a corresponding system of functional equations we obtain:
\begin{eqnarray*}
  H(z) &=& K(z) + z H(z),\\
  K(z) &=& z K(z) \gL + D(z)
\end{eqnarray*}
and hence
\begin{eqnarray}\label{eq:H-hnf}
  \nonumber K(z) &=& \frac{D(z)}{1 - z \gL},\\
  H(z) &=&  \frac{K(z)}{1-z}.
\end{eqnarray}
It is easy to verify that we have
\begin{eqnarray}\label{eq:K-nhnf}
  K(z) &=& z +z \gL.  
\end{eqnarray}

Naturally, the above equation suggests an appropriate translation between the set of neutral head normal forms and the set of plain $`l$-terms. Consider the following partial mapping $\K \mapsto \cgL$:
\begin{align*}
\zerodot\, N_1 N_2\ldots N_m \, &\longleftrightarrow\, (`l\, N_1) N_2\ldots N_m & \text{where $m > 0$}\\
(S \underbar{\sf n})\, N_1 \ldots N_m \, &\longleftrightarrow\, \underbar{\sf n}\, N_1 \ldots N_m & \text{where $m \geq 0$}
\end{align*}

Note that the neutral head normal forms are of size by one greater than the size of their plain $`l$-term counterparts. Since each plain $`l$-term is either in form of $(`l\, N_1) N_2\ldots N_m$ for some $m > 0$ or $\underbar{\sf n} \, N_1 \ldots N_m$ (note that in this case $m$ can be equal to $0$), the above mapping is surjective, explaining the $z \gL$ part in~\autoref{eq:K-nhnf}. The $z$ part comes from the fact that the only $`l$-term in neutral head normal form not covered by the mapping is $\zerodot$, which is of size one.

Immediately, from~\autoref{th:Loo-approx} we get the following results.
\begin{prop}
  The asymptotic approximation of the number of $`l$-terms in neutral head normal form of size $n+1$ is given by
  \begin{displaymath}
    [z^{n+1}]K(z) \sim \left(\frac{1}{\domsing{L_\infty}}\right)^n \frac{C}{n^{\frac{3}{2}}}
  \end{displaymath}
  with $C \doteq 0.60676$ and $\domsing{L_\infty} \doteq 0.29559$.
\end{prop}
In particular, we obtain the following easy consequence.
\begin{corollary}
  The density of neutral head normal forms in the set of plain terms equal to $\domsing{L_\infty}$.
\end{corollary}

Solving~\autoref{eq:H-hnf} we can find the asymptotic approximation on the growth rate of head normal forms, similarly to plain $`l$-terms (see~\autoref{th:Loo-approx}).

\begin{prop}
The asymptotic approximation of the number of $`l$-terms in head normal form of size $n$ is given by
   \begin{displaymath}
    [z^n]H(z) \sim \left(\frac{1}{\domsing{L_\infty}}\right)^n \frac{C_H}{n^{\frac{3}{2}}}
  \end{displaymath}
  with $ C_H \doteq 0.254625911836762946$. 
\end{prop}
\begin{proof}
  The proof is analogous to the one of~\autoref{th:Loo-approx} with
  \begin{displaymath}
    C_H = \frac{-\sqrt{\domsing{L_\infty} \frac{Q(\domsing{L_\infty})}{1-\domsing{L_\infty}}}}{2(1-\domsing{L_\infty})`G(-\frac{1}{2})}
    \doteq 0.254625911836762946.
  \end{displaymath}
\end{proof}

Comparing it with the growth rate of $\nthcoeff{\gL}$ we obtain the following corollary.

\begin{corollary}
The density of head normal forms in the set of plain terms equal to \[\frac{\domsing{L_\infty}}{1-\domsing{L_\infty}} \doteq 0.41964337760707887.\]
\end{corollary}

Note that with the above density results, we are able to explain the effectiveness of Boltzmann samplers for plain $`l$-terms (see, e.g.~\cite{gry_les}), used with an additional \emph{rejection} phase. Consider the following approach. In order to sample a (neutral) head normal $`l$-term, we draw random plain $`l$-terms until the first (neutral) head normal $`l$-term is sampled. In the case of head normal forms, the expected number of samples for large $n$ equals $\frac{1-\domsing{L_\infty}}{\domsing{L_\infty}} \doteq 2.383$, while in the case of neutral head normal forms it is equal to $\frac{1}{\domsing{L_\infty}} \doteq 3.383$.

\subsection{Counting terms with bounded number of free indices}\label{subsec:m-open-lambda-terms}

In this section we are interested in counting terms with bounded number of
    distinct free de~Bruijn indices. We start by giving the generating function $D_m(z)$ 
    associated with the set of first $m$ indices.

    \begin{prop}
    Let $\D_m = \set{\underbar{\sf 0},\underbar{\sf 1},\ldots,\underbar{\sf{m-1}}}$
    where $m \in \nats$. Then
        \[ D_m(z) = \frac{z(1-z^m)}{1-z} .\]
    \end{prop}

    \begin{proof}
        Let us notice that 
        \[ \nthcoeff{D_m(z)} = 
              \begin{cases}
              1 &\mbox{if }  1 \leq n \leq m,\\
              0 &\mbox{otherwise.}
              \end{cases}
        \]
        It follows that we can express $D_m(z)$ as $D(z) - z^m D(z)$. Using~\autoref{prop:D-gen-fun} we finally obtain $D_m(z) = \frac{z}{1-z} -
        \frac{z^{m+1}}{1-z} =\frac{z(1-z^m)}{1-z}$, finishing the proof.
    \end{proof}
    
	Let $m \in \nats$. We denote by $\cgL[m]$ the set of $`l$-terms whose free indices are 
	elements of $\D_m$. Obviously, for every $m$ we have $\cgL[m] \subseteq \cgL[m+1]$.

    \begin{prop}
    The generating function associated with the set $\cgL[m]$ is given by
        \[\gL[m] = \frac{1-\sqrt{1 - 4z^2\left(\gL[m+1]
        + \frac{1-z^m}{1-z}\right)}}{2z}.\]
    \end{prop}

    \begin{proof}
		Due to the structure of $`l$-terms, we can set the following specification defining $\cgL[m]$:
        \[\cgL[m] = \cgL[m] \cgL[m] `(+) `l \cgL[m+1] 
        `(+)\, \D_m,\]
        which immediately implies
        \[\gL[m] =  z {\gL[m]}^2 - z \gL[m+1] + \frac{z(1-z^m)}{1-z}.\]
        Solving the above equation in $\gL[m]$, we obtain
        \[\gL[m] = \frac{1-\sqrt{\DeltagLm}}{2z} = \frac{1-\sqrt{1 -
        4z^2\left(\gL[m+1] + \frac{1-z^m}{1-z}\right)}}{2z}.\]
    \end{proof}
    
    Notice that $\gL[m]$, and in particular $\gL[0]$ -- counting the number of
    closed $`l$-terms, is defined using $\gL[m+1]$. If this definition is
    developed, then $\gL[m]$ is expressed by means of infinitely nested
    radicals~--~a known phenomenon already observed in other models of 
    $`l$-calculus (see, e.g.~\cite{gry_les,
    gittenberger-2011-ltbuh}).
    
    \medskip
	Although the challenging problem of finding asymptotic approximations on the number of closed $`l$-terms is still open, in~\cite{gittenberger_et_al:LIPIcs:2016:5741}, Gittenberger and Gołębiewski give the following bounds on the asymptotic growth rate of $\seq{\nthcoeff{\gL[0]}}_{n\in\nats}$.
    
    \begin{theorem}[see~\cite{gittenberger_et_al:LIPIcs:2016:5741}, Lemma 14]
    The following bounds hold:
    \[ \liminf_{n\to\infty} \frac{\nthcoeff{\gL[0]}}{\underline{C}n^{\nicefrac{3}{2}} \domsing{L_\infty}^{-n}} \geq 1 \quad \text{and} \quad \limsup_{n\to\infty} \frac{\nthcoeff{\gL[0]}}{\overline{C}n^{\nicefrac{3}{2}} \domsing{L_\infty}^{-n}} \leq 1,\]
    where $\underline{C} \doteq 0.07790995266$ and $\overline{C} \doteq 0.07790998229$.
    \end{theorem}
    
    The above result implies, inter alia, that the asymptotic density of closed $`l$-terms in the set of plain ones cannot be equal to zero. Comparing the obtained constants $\underline{C}$ and $\overline{C}$ with the constant $C \doteq 0.60676$ in the asymptotic approximation of plain $`l$-terms (see~\autoref{th:Loo-approx}) we get the following corollary.
    
    \begin{corollary}
    We have the following numerical bounds on the lower and upper density of closed $`l$-terms in the set of plain ones:
    \begin{eqnarray*}
    \liminf_{n\to\infty} \frac{\nthcoeff{\gL[0]}}{\nthcoeff{\gL}} &\geq& 0.1284032445447953,\\
    \limsup_{n\to\infty} \frac{\nthcoeff{\gL[0]}}{\nthcoeff{\gL}} &\leq& 
    0.1284032933779419.\\
    \end{eqnarray*}
    \end{corollary}
    
In other words for large $n$, we should expect that in the set of $`l$-terms of size $n$, roughly $12,84\%$ of them are closed. Immediately, this suggests the following naive approach for a dedicated rejection sampler for closed $`l$-terms: draw random plain $`l$-terms until the first closed one is sampled. With the above relative density bounds, we expect that in order to draw a uniformly random closed $`l$-term, we have to repeat the sampling roughly $13$ times on average, before the first success.

\section{Counting $`l$-terms with another notions of size}\label{sec:another-size-notions}
    Assume we take another notion of size in which $\zerodot$ has size zero,
    applications are of size two, whereas abstraction and successor keep their
    original size one. Formally,
    \begin{eqnarray*}
    		|`l \, N| &=& |N| + 1,\\
        |N \, M|&=& |N| + |M| + 2,\\
        |S \underbar{\sf n}| &=& |\underbar{\sf n}| + 1,\\
        |\zerodot| &=& 0.
    \end{eqnarray*}
    It is easy to verify that the corresponding generating 
    function\footnote{We write this function $A_1(z)$
    as a reference to the function $A(x,1)$ described in 
    \href{http://oeis.org/A105632}{\textbf{A105632}} of the \emph{Online
    Encyclopedia of Integer Sequences}~\cite{oeisEncyclopedia}.} $A_1$ fulfills
    the identity
    \[A_1(z) = z^2 A_1^2(z) - (1-z) A_1(z) + \frac{1}{1-z}.\]
   	
   	In particular, we have $\gL = z\,A_1(z)$ and hence $\nthcoeff{A_1(z)} = [z^{n+1}]\gL$. Indeed, the number of zeros in an arbitrary $`l$-term $T$ is equal to the number of its applications plus one. Suppose that the number of applications in $T$ is equal to $d$. Then, in the natural size notion where each constructor is of size one, applications and zeros in $T$ contribute $2d + 1$ to its size. On the other hand, in the above size notion applications and zeros contribute just $2d$ to $T$'s size. Since both size functions set the size of abstractions and successors to one, we obtain $\nthcoeff{A_1(z)} = [z^{n+1}]\gL$. It follows that both notions of size yield the sequence \href{http://oeis.org/A105633}{\textbf{A105633}}.

Suppose that we assume another size notion where:
  \begin{eqnarray*}
    |`l \, N| &=& |N| + 1,\\
    |N\,M|&=& |N| + |M| + 1,\\
    |S \underbar{\sf n}| &=& |\underbar{\sf n}| + 1,\\
    |\zerodot| &=& 0.
  \end{eqnarray*}
Then, the corresponding generating function $\agL$ is the solution of
\begin{displaymath}
  z \agL^2 - (1-z) \agL + \frac{1}{1-z} = 0 
\end{displaymath}
with discriminant $`D_{\agL} = \frac{1 - 7z + 3z^2 - z^3}{1-z}$ yielding the dominating singularity $`r_{M_\infty} \doteq 0.152292401860433$ and $1/`r_{M_\infty} \doteq 6.5663157700831193$. The first $10$ values of ${\seq{\nthcoeff{\agL}}}_{n\in\nats}$ are:

\begin{center}
\medskip
\begin{small}
  1, 3, 10, 40, 181, 884, 4539, 24142, 131821,
 734577, 4160626 .
 \end{small}
\medskip
\end{center}
This sequence is known as \href{http://oeis.org/A258973}{\textbf{A258973}} in the \emph{Online Encyclopedia of Integer Sequences}~\cite{oeisEncyclopedia} and grows significantly faster than \href{http://oeis.org/A105633}{\textbf{A105633}}. 

Remarkably, under some additional technical assumptions on the constructor sizes, counting sequences of plain $`l$-terms yield similar asymptotic expansions and behaviour. We refer the curious reader to~\cite{gittenberger_et_al:LIPIcs:2016:5741}.

\section{Counting $\lambda$-terms containing fixed $\lambda$-terms as subterms}\label{sec:fixed-lambda-terms}

Let $M$ be an arbitrary $`l$-term of size $p$ and $\T_M$ denote the set of
    $`l$-terms that contain $M$ as a subterm. In this section we focus on the
    asymptotic density of $\T_M$ in the set of all $`l$-terms.
    
    \begin{theorem}
        For a fixed term $M$, the asymptotic density of $\T_M$ is equal to $1$. 
        In other words, asymptotically almost all $`l$-terms contain $M$ as a subterm.
    \end{theorem}

    \begin{proof}
        Consider an arbitrary $T \in \T_M$. Either $T$ is equal to $M$, or
        $M$ is a proper subterm of $T$. In the latter case we have four
        additional cases. Either $T$ is an abstraction, or $T = T_1 T_2$
        and $M$ is a subterm of $T_1$, or $T_2$, or both. Combining, we obtain the
        following equation:
        \begin{eqnarray*}
            \T_M &=& M `(+) `l\, \T_M `(+) \T_M\, \cgL `(+) \cgL\, 
                \T_M `(-) \T_M\, \T_M.
        \end{eqnarray*}
        
        Note that by adding $\T_M\cgL$ and $\cgL\T_M$ we count each
        term $T = T_1 T_2$ containing $M$ in both $T_1$ and $T_1$ twice,
        therefore we have to subtract $\T_M\, \T_M$. Such a representation yields the
        following functional quadratic equation involving the corresponding 
        generating function $T_M(z)$:
        \[T_M(z) = z^p + z\,T_M(z) + 2z\, T_M(z)\,
            \gL - z\,T^{2}_M(z).\]
    
        Since $\sqrt{\Delta_{\gL}} = 1 - 2z\, \gL - z$ (see
        ~\autoref{prop:Loo-gen-fun}), we can express the discriminant of $T_M(z)$ as
        $`D_{T_M(z)} = \Delta_{\gL} + 4 z^{p+1}$. Hence $`D_{T_M(z)} > `D_{\gL}$.
        It follows that the root $`r_{T_M}$ of smallest modulus of $`D_{T_M(z)}$ 
        is strictly larger than the root $`r_{L_\infty}$ of smallest modulus of 
        $`D_{\gL}$, i.e.~$`r_{T_M} >`r_{L_\infty}$. Moreover,
        $T_{M}(z) = \frac{\sqrt{`D_{T_M(z)}} - \sqrt{`D_{\gL}}}{2z}$ and thus
        the generating function counting the number of $`l$-terms which do
        not contain $M$ as a subterm is given by
        \[ \gL - T_M(z) = \frac{(1-z) - \sqrt{`D_{T_M(z)}}}{2z}. \]
    
        Applying~\autoref{thm:exponential-growth-formula} we immediately
        get that the above set has asymptotic density $0$ and thus $\T_M$
        has asymptotic density equal to $1$.
    \end{proof}
    
    \begin{corollary}
        Asymptotically almost every $`l$-term is neither strongly normalising, nor typeable, nor in normal form.
    \end{corollary}

    \begin{proof}
            Consider the aforementioned $\mathsf{\Omega}$. Clearly, it is neither typeable nor in normal form. Moreover, as it is not
            normalising and asymptotically almost all $`l$-terms contain it as a
            subterm, we immediately get our claim.
    \end{proof}
   
    Let us notice the striking discrepancy between the density of strongly
    normalising terms in the natural model and the corresponding density in the model
    considered in~\cite{david}. In the latter case, variables tend to be arbitrarily
    far from their binders, since they do not contribute to the overall size. In the
    natural model, however, increasing an index (i.e., increasing the distance of the
    variable from the binder) increases the overall size and thus indices tend to be
    rather near their binding lambdas.

\section{Conclusions}\label{sec:conclusions}
We investigated the combinatorial aspects of $`l$-terms in the model with unary de~Bruijn indices and natural size notion. We provided effective size-preserving translations among plain $`l$-terms, black-white trees and zigzag-free ones. By exhibiting a bijection between Motzkin trees and neutral forms, we showed that our translation allows to exploit the exact-size Motzkin tree sampler of Bacher et al.~\cite{BBJ} yielding an exact-size sampler for neutral $`l$-terms. Next, we considered the classes of head normal forms and neutral head normal forms, linking their positive densities in the set of plain $`l$-terms with the effectiveness of rejection Boltzmann samplers for the aforementioned classes. Finally, we proved that strongly normalising terms, as typeable ones or normal forms, are asymptotically negligible in the set of all $`l$-terms, contrary to the model considered in~\cite{david}. The following figure summarises our density results.
\begin{figure}[!htp]
    \centering
    \begin{tabular}{c@{\qquad\qquad}c@{\qquad}c@{\qquad\qquad\qquad\qquad}c}
      \textsf{nf}&\textsf{nhnf}&& $\T_M$\\
       \textsf{sn}&&\textsf{hnf} &$\mathsf{\overline{sn}}$\\[2pt]\hline
      \large{\textsf{0}} &\large{\textsf{0.295...}}& \large{\textsf{0.419...}} & \large{\textsf{1}}
    \end{tabular}
\begin{eqnarray*}
 \textsf{nf} &-& \text{normal forms} \\
 \T_M &-& \text{terms containing subterm } M\\
 \textsf{nhnf} &-& \text{neutral head normal forms}\\
 \textsf{hnf} &-& \text{head normal forms}\\
 \textsf{sn} &-& \text{strongly normalising terms}\\
  \overline{\textsf{sn}} &-& \text{not strongly normalising terms}
\end{eqnarray*}
    \label{fig:summary}
  \end{figure}
  \vspace{-0.9cm}
\bibliographystyle{plain}
\bibliography{references}
\end{document}